\documentclass[a4paper,12pt]{article}

\usepackage{amssymb}
\usepackage{amsmath, amssymb}
\usepackage{amsthm}
\usepackage{graphicx}
\usepackage{hyperref}

\newcommand{\carka}{,\penalty0\relax}
\newcommand{\zlom}{\penalty0\relax}

\newtheorem{theorem}{Theorem}

\newtheorem{lemma}[theorem]{Lemma}
\newtheorem{proposition}[theorem]{Proposition}
\newtheorem{corollary}[theorem]{Corollary}

\def\email#1{\href{mailto:#1}{\texttt{#1}}}
\def\NP{\ensuremath{\mathcal{NP}}}
\def\Z{\ensuremath{\mathbb{Z}}}
\def\N{\ensuremath{\mathbb{N}}}

\title{Packing Chromatic Number of Distance Graphs}
\author{Jan Ekstein$^{*}$ ~~~ P\v{r}emysl Holub\thanks{University of West Bohemia, Pilsen, Czech Republic,\newline
      e-mail: \email{$\lbrace$ekstein, holubpre$\rbrace$@kma.zcu.cz}, supported by grants 1M0545,
      MEB 101014 of the Czech Ministry of Education.}\and
      Bernard Lidick\'y\thanks{Charles University, Prague, Czech Republic, \newline
      e-mail: \email{bernard@kam.mff.cuni.cz}\carka
      ~supported by GAUK 60310, GACR~201/09/0197 of~the Czech Ministry of Education.}}
\date{\today}

\begin{document}

\maketitle

\begin{abstract}
  The packing chromatic number $\chi_{\rho}(G)$ of a graph $G$ is the smallest integer $k$ such that vertices of $G$ can
  be partitioned into disjoint classes $X_{1}, ..., X_{k}$ where vertices in $X_{i}$ have pairwise distance greater than
  $i$. We study the packing chromatic number of infinite distance graphs $G(\Z, D)$, i.e. graphs with the set $\Z$ of
  integers as vertex set and in which two distinct vertices $i, j \in \Z$ are adjacent if and only if $|i - j| \in D$.

  In this paper we focus on distance graphs with $D = \{1, t\}$. We improve some results of Togni who initiated the
  study. It is shown that $\chi_{\rho}(G(\Z, D)) \leq 35$ for sufficiently large odd $t$ and
  $\chi_{\rho}(G(\Z, D)) \leq 56$ for sufficiently large even $t$. We also give a lower bound 12 for $t \geq 9$ and
  tighten several gaps for $\chi_{\rho}(G(\Z, D))$ with small $t$.
\end{abstract}

\section{Introduction}
 In this paper we consider simple undirected graphs only. For terminology and notations not defined here we refer to
 \cite{Bon}. Let $G$ be a connected graph and let $\mbox{dist}_{G}(u, v)$ denote the distance between vertices $u$ and
 $v$ in $G$. We ask for a partition of the vertex set of $G$ into disjoint classes $X_{1}, ..., X_{k}$ according to the
 following constraints. Each color class $X_{i}$ should be an \emph{$i$-packing}, a set of vertices with property that
 any distinct pair $u, v \in X_{i}$ satisfies $\mbox{dist}_{G}(u, v) > i$. Such a partition is called a
 \emph{packing $k$-coloring}, even though it is allowed that some sets $X_{i}$ may be empty. The smallest integer $k$
 for which there exists a packing $k$-coloring of $G$ is called the \emph{packing chromatic number} of $G$ and it is
 denoted $\chi_{\rho}(G)$. The very first results about packing chromatic number were obtained by
 Slopper~\cite{bib-sloper04}. He studied an \emph{eccentric coloring} but his results were directly translated to the
 packing chromatic number. The concept of packing chromatic number was introduced by Goddard et al. \cite{God} under the
 name \emph{broadcast chromatic number}.
 The term packing chromatic number was later proposed by Bre\v{s}ar et al.~\cite{Bre}. The determination of the packing
 chromatic number is computationally difficult. It was shown to be \NP-complete for general graphs in~\cite{God}. Fiala
 and Golovach \cite{FiaGol} showed that the problem remains \NP-complete even for trees.

 The research of the packing chromatic number was driven by investigating $\chi_{\rho}(\Z^2)$ where $\Z^2$ is the
 Cartesian product of two infinite paths - the (2-dimensional) \emph{square lattice}. Goddard et al. \cite{God} showed
 that $9 \leq \chi_{\rho}(\Z^2) \leq 23$. Fiala et al. \cite{bib-fiala09+} improved  the lower bound to 10 and Holub
 and Soukal \cite{HolSo} improved the upper bound to 17. The lower bound was pushed further to 12 by Ekstein et
 al.~\cite{bib-ekstein10+}. For  $\mathbb{Z}^{3}$ see~\cite{bib-fiala09+,bib-finbow07+}.

 Let $D = \{d_{1}, d_{2}, ...\carka d_{k}\}$, where $d_{i}$ are positive integers and $i = 1, 2, ..., k$. The (infinite)
 \emph{distance graph} $G(\mathbb{Z} ,D)$ with distance set $D$ has the set $\mathbb{Z}$ of integers as a vertex set and
 in which two distinct vertices $i, j \in \mathbb{Z}$ are adjacent if and only if $|i - j| \in D$. We denote the graph
 $G(\mathbb{Z}, \{a, b\})$ by $D(a, b)$. The study of a coloring of distance graphs was initiated by Eggleton et al.
 \cite{Eggl}. In last twenty years there were more than 60 papers concerning this topic. We recall e.g. contributions by
 Voigt and Walter \cite{VoiWal}, Ruzsa et al. \cite{RuzTuVoi}, Liu \cite{Liu}, Liu and Zhu \cite{LiuZhu} and Barajas and
 Serra \cite{BaSe}.

 \begin{table}[ht]
 \begin{center}
  \begin{tabular}{|c|c|c|c|}
   \hline
   D & $\chi_{\rho} \geq$ & $\chi_{\rho} \leq$ \\
   \hline
   1,2 & 8       & 8 \\ 
   1,3 & 9       & 9 \\ 
   1,4 & 11   &  16 \\ 
   1,5 & 10   & 12 \\ 
   1,6 & 11   & 23 \\ 
   1,7 & 10   & 15 \\ 
   1,8 & 11   & 25 \\ 
   1,9 & 10   & 18 \\ 
   \hline
   \end{tabular}
 \hskip 1em
   \begin{tabular}{|c|c|c|c|c|}
   \hline
   D &  $\chi_{\rho} \geq$  &$\chi_{\rho} \leq$ \\ 
   \hline
   1,2 &  8        &  8      \\
   1,3 &  9        &  9      \\
   1,4 &  {\bf 14} & {\bf 15}\\
   1,5 &  {\bf 12} &  12     \\
   1,6 &  {\bf 15} &  23     \\
   1,7 &  {\bf 14} &  15     \\
   1,8 &  {\bf 15} &  25     \\
   1,9 &  {\bf 13} &  18     \\
   \hline
  \end{tabular}
   \caption{Lower and upper bounds for the packing chromatic number of $D(1, t)$.
   Left table contains previously known bounds and the right table contains current
   bounds.}
   \label{Tabulka}
   \end{center}
 \end{table}

 The study of a packing coloring of distance graphs was initiated by Togni~\cite{Togni}. Results for $D(1, t)$ for small
 values of $t$, obtained by Togni~\cite{Togni}, are summarized in the left part of Table~\ref{Tabulka}. Our improvements
 are \zlom emphasized in the right part of the table and they were obtained by a computer. We wrote two independent
 programs (one in Pascal and other one in C++). The source codes and the outputs of the programs can be downloaded from
 \url{http://kam.mff.cuni.cz/~bernard/dist}.

 \vskip 1.5cm

 For larger $t$ Togni proved the following theorem.

 \begin{theorem}\emph{\textbf{\cite{Togni}}}
  \label{PakovDist} For every $q,t \in \N$:
  $$\chi_{\rho}(D(1, t)) \leq \left\{%
  \begin{array}{ll}
    86  & \mathrm{if~} t = 2q + 1, q \geq 36, \\
    40  & \mathrm{if~} t = 2q + 1, q \geq 223, \\
    173 & \mathrm{if~} t = 2q, q \geq 87, \\
    81  & \mathrm{if~} t = 2q, q \geq 224, \\
    29  & \mathrm{if~} t = 96q \pm 1, q \geq 1, \\
    59  & \mathrm{if~} t = 96q +1 \pm 1, q \geq 1. \\
  \end{array}%
  \right.    $$
 \end{theorem}

 We improve some results of Theorem \ref{PakovDist} as follows.

 \begin{theorem}
  \label{DistanceGraph}
  For any odd integer $t \geq 575$, $$\chi_{\rho}(D(1, t)) \leq 35.$$
  For any even integer $t \geq 648$, $$\chi_{\rho}(D(1, t)) \leq 56.$$
 \end{theorem}

 We also give a lower bound for the packing chromatic number of $D(1, t)$ for $t \geq 9$, as a corollary of the
 following statement.

 \begin{theorem}\emph{\textbf{\cite{bib-ekstein10+}}}
 \label{CtvercovaMrizka}
  The packing chromatic number of the square lattice is at least 12.
 \end{theorem}

 \begin{corollary}
  \label{Dolni mez}
  Let $D(1, t)$ be a distance graph, $t \geq 9$ an integer. Then $$\chi_{\rho}(D(1, t)) \geq 12.$$
 \end{corollary}

 Throughout the rest of the paper by a coloring we mean a packing \zlom coloring.

%%%%%%%%%%%%%%%%%%%%%%%%%%%%%%%%%%%%%%%%%%% Small t
 \section{$D(1, t)$ with small $t$}
 \label{Small t}
  In this section we prove new lower and upper bounds for the packing chromatic number of $D(1, t)$  which are mentioned
  in Table~\ref{Tabulka}.

  \begin{lemma}
   \label{lem-15}
    $\chi_{\rho}(D(1, 4)) \leq 15.$
  \end{lemma}
  \begin{proof}
   We prove this lemma by exhibiting a repeating pattern for 15-packing coloring of $D(1,4)$.
   The pattern has period 320 and is given here:

   \noindent
   1,3,1,2,4,1,5,1,8,2,1,3,1,10,11,1,2,1,6,4,1,3,1,2,5,1,7,1,9,2,1,3,1,12,4,1,2,1,13,\\
   8,1,3,1,2,6,1,5,1,4,2,1,3,1,7,10,1,2,1,15,14,1,3,1,2,5,1,4,1,11,2,1,3,1,6,9,1,2,1,\\
   8,7,1,3,1,2,4,1,5,1,12,2,1,3,1,10,13,1,2,1,4,6,1,3,1,2,5,1,7,1,8,2,1,3,1,4,14,1,2,\\
   1,11,9,1,3,1,2,6,1,5,1,4,2,1,3,1,7,10,1,2,1,8,12,1,3,1,2,5,1,4,1,13,2,1,3,1,6,9,1,\\
   2,1,15,7,1,3,1,2,4,1,5,1,8,2,1,3,1,10,11,1,2,1,6,4,1,3,1,2,5,1,7,1,9,2,1,3,1,12,4,\\
   1,2,1,13,8,1,3,1,2,6,1,5,1,4,2,1,3,1,7,10,1,2,1,14,15,1,3,1,2,5,1,4,1,11,2,1,3,1,\\
   6,9,1,2,1,8,7,1,3,1,2,4,1,5,1,12,2,1,3,1,10,13,1,2,1,4,6,1,3,1,2,5,1,7,1,8,2,1,3,1,\\
   4,11,1,2,1,15,9,1,3,1,2,6,1,5,1,4,2,1,3,1,7,10,1,2,1,8,12,1,3,1,2,5,1,4,1,13,2,1,\\
   3,1,6,9,1,2,1,14,7.

   The pattern was found with help of a computer using simulated annealing heuristics \cite{LaaAart}.
  \end{proof}

  \begin{lemma}
   \label{lem-small-t-computer}
    $$14 \leq \chi_{\rho}(D(1, 4)),$$
    $$12 \leq \chi_{\rho}(D(1, 5)),$$
    $$14 \leq \chi_{\rho}(D(1, 7)),$$
    $$13 \leq \chi_{\rho}(D(1, 9)).$$
  \end{lemma}
  \begin{proof}
   These results were obtained by a computer using a brute force search programs. We have written two independent
   programs (one in Pascal and one in C++) implementing the brute force search. The programs take vertices
   $X = \{1,2,\ldots k\}$ from $D(1,t)$. Then they try to construct a packing coloring~$\varrho$ of $X$ using colors
   from 1 up to $c$. First, they assign $\varrho(1) = c$ and then they try to extend $\varrho$ to $X$. If the extension
   is not possible we conclude that $\chi_{\rho}(D(1,t)) > c$. The results of computations are summarized
   in Table~\ref{tab-comsearch}.

 \begin{table}[ht]
 \begin{center}
  \begin{tabular}{|c|c|c|c|c|}
   \hline
   D   & $c$ & $k$ & Configurations & Time  \\
   \hline
   1,4 &  13  & 81  & $6.4 \cdot 10^{12}$ & 26 days \\
   1,5 &  11  & 134 & $8.1 \cdot 10^{9} $ & 25 minutes   \\
   1,7 &  13  & 229 & $6.9 \cdot 10^{13}$ & 335 days   \\
   1,9 &  12  & 66  & $6.2 \cdot 10^{12}$ & 28 days \\
   \hline
   \end{tabular}
   \caption{Computations from Lemma~\ref{lem-small-t-computer}. Time of the computation is measured on a workstation
            from year 2010.}
   \label{tab-comsearch}
   \end{center}
 \end{table}
  \end{proof}

  Let $H_k$ denote a finite subgraph of $D(1,t)$ on vertices $1,\dots, \, k$ and let $H'_k$ denote a finite subgraph of
  $D(1,t)$ on vertices $-k,-k+1,\dots, \, k$.

 For a subset $X$ of vertices of $D(1,t)$ we define its {\em density} $d(X)$ as
 \begin{displaymath}
    d(X) =\limsup\limits_{k\to\infty} \frac{\vert X \cap V(H'_k)\vert}{\vert V(H'_k)\vert}.
 \end{displaymath}
 For a color $c$ we define its \emph{density} $d(c)$ as
 \begin{displaymath}
   d(c)= \max\limits_{\chi} d(X_c),
 \end{displaymath}
 where $\chi$ is a packing coloring of $D(1,t)$ and $X_c$ is a $c$-packing.
 Similarly, by $d(c_1,\dots ,c_l)$ we mean
 \begin{displaymath}
   d(c_1,\dots ,c_l)= \max\limits_{\chi} d(X_{c_{1}} \cup \ldots \cup X_{c_{l}}).
 \end{displaymath}

The following statement was proved in \cite{bib-fiala09+}.

 \begin{lemma}\emph{\textbf{\cite{bib-fiala09+}}}
  \label{lem-delic}
  If there exists a coloring of $D(1,t)$ by $k$ colors then,
  for every $1 \leq l \leq k$, it holds that
  \begin{displaymath}
   \sum\limits_{i=1}^k d(i)  \geq  d(1,\dots, l) + \sum\limits_{i=l+1}^k d(i) \geq  d(1,\dots, k) = 1.
  \end{displaymath}
 \end{lemma}

  \begin{lemma}
   \label{Male t}
    $$15 \leq \chi_{\rho}(D(1, 6)),$$
    $$15 \leq \chi_{\rho}(D(1, 8)).$$
  \end{lemma}

 \begin{proof}
 To the contrary we suppose that $\chi_{\rho}(D(1,6)) \leq 14$. Using a computer we verified that
 $d( 1, 2, 3, 4) \leq \frac {31}{41}$ since we can color at most $31$ vertices of $H_{41}$ using colors
 $1, 2, 3, 4$. The computation took about three minutes and it checked $4.6 \cdot 10^{9}$ configurations. Clearly,
 $d(i) \leq \frac{1}{6i-9}$ for $i\geq 2$ since there is no pair of vertices in $H_{6i-9}$ with distance greater
 than $i$ and hence at most one vertex of $H_{6i-9}$ can be colored by color $i$. By Lemma \ref{lem-delic} we
 easily get
 \begin{displaymath}
  d(1, 2, \dots, 14) \leq d( 1, 2, 3, 4) + \sum\limits_{i=5}^{14} d(i) \leq \frac{31}{41} + \frac{1}{21} + \dots +
  \frac{1}{75} = 0.999771 < 1,
 \end{displaymath}
 which is not possible since $d(1, 2, \dots, 14) = 1$ by the assumption that \zlom $\chi_{\rho}(D(1,6))~\leq~14$.

 Now to the contrary we suppose that $\chi_{\rho}(D(1,8)) \leq 14$. Using a computer we verified that
 $d( 1, \dots, 6) \leq \frac {50}{58}$ since we can color at most $50$ vertices of $H_{58}$ using colors
 $1, \dots, 6$. The computation took about sixty hours and it checked $7.5 \cdot 10^{11}$ configurations. Clearly,
 $d(i) \leq \frac{1}{8i-20}$ for $i\geq 3$ since there is no pair of vertices in $H_{8i-20}$ with distance greater
 than $i$ and hence at most one vertex of $H_{8i-20}$ can be colored by color $i$. By Lemma \ref{lem-delic} we
 easily get
 \begin{displaymath}
  d(1, 2, \dots, 14) \leq d( 1, \dots, 6) + \sum\limits_{i=7}^{14} d(i) \leq \frac{50}{58} + \frac{1}{36} + \dots +
  \frac{1}{92} = 0.999110 < 1,
 \end{displaymath}
 which is not possible since $d(1, 2, \dots, 14) = 1$ by the assumption that \zlom $\chi_{\rho}(D(1,8))~\leq~14$.
 \end{proof}

%%%%%%%%%%%%%%%%%%%%%%%%%%%%%%%%%%%%%%%%%%% Large t
 \section{$D(1, t)$ with large $t$}
 \label{Large t}

 A key observation for this section is that a distance graph $D(1, t)$, for $t > 1$, can be drawn as an infinite spiral
 with $t$ lines orthogonal to the spiral (e.g. $D(1, 5)$ on Figure~\ref{D(1,5)}).

  \begin{figure}[ht]
   \begin{center}
   \includegraphics[width=12.cm]{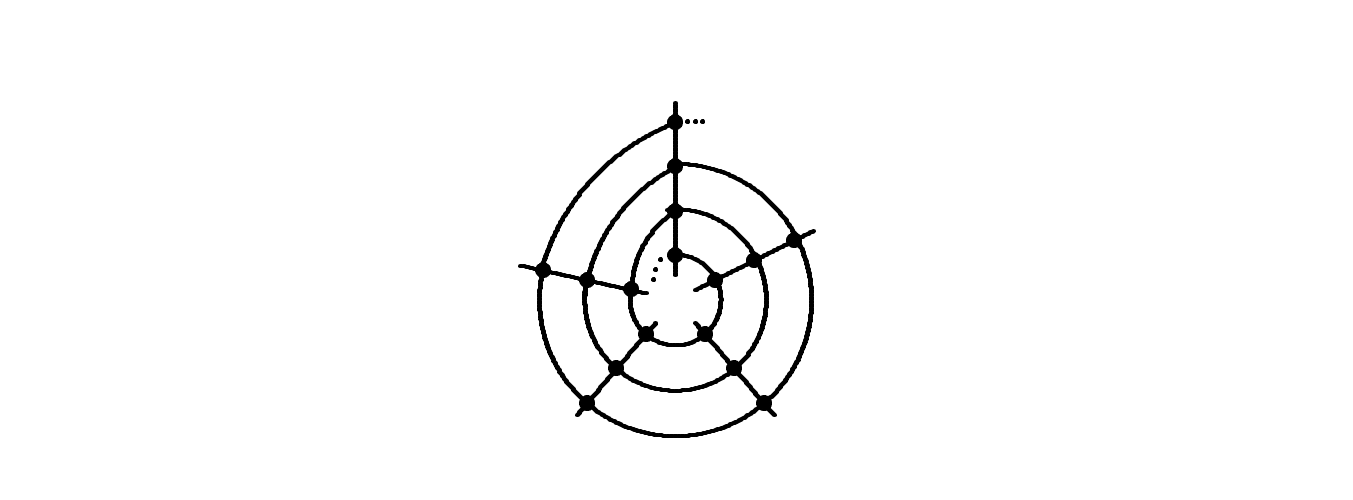}\\
   \end{center}
  \caption{Distance graph $D(1,5)$.}
  \label{D(1,5)}
 \end{figure}

 For $i \in \{0, 1, ..., t - 1\}$, the \emph{$i$-band} in a distance graph $D(1, t)$, denoted by $B_{i}$, is an infinite
 path in $D(1, t)$ on the vertices $V(B_{i}) = \{i + kt, k \in \mathbb{Z}\}$. Note that the band $B_{i}$ corresponds to
 one of $t$ lines orthogonal to the spiral. For $i \in \{0, 1, ..., t - 24\}$, the \emph{$i$-strip} in a distance graph
 $D(1, t)$, $t > 23$, denoted by $S_{i}$, is a subgraph of $D(1, t)$ induced by the union of vertices of
 $B_{i}, B_{i + 1}, ..., B_{i + 23}$.

 We use the following statement proved by Goddard et al. in \cite{God}.

 \begin{proposition}\emph{\textbf{\cite{God}}}
  \label{Goddard}
   For every $k \in \N$, the infinite path can be colored by colors $k, k + 1, ..., 3k + 2$.
 \end{proposition}

 Holub and Soukal \cite{HolSo} improved the upper bound for a packing coloring of the square lattice to 17 by finding
 a pattern on 24$\times$24 vertices using color 1 on positions as white places on a chessboard. We use this pattern to
 prove the following lemma.

 \begin{lemma}
  \label{Strip}
   Let $D(1, t)$ be a distance graph, $t > 24$, and $S_{i}$ its $i$-strip. Then $\chi_{\rho}(S_{i}) \leq 17$.
 \end{lemma}

 \begin{proof}
  We cyclically use the pattern on 24$\times$24 vertices to color all the vertices of $S_{i}$. Hence it is obvious that
  $\chi_{\rho}(S_{i}) \leq 17$.
 \end{proof}

 \begin{lemma}
  \label{lem-band}
   Let $D(1, t)$ be a distance graph and $B_{i}$ its $i$-band. If vertices $\{i + 2jt, j \in \Z\}$ are colored by
   color 1, then it is possible to extend the coloring to all vertices of $B_{i}$ using colors $k, k +1, ..., 2k - 1$,
   for every $k \in \N, k > 2$.
 \end{lemma}

 \begin{proof}
  We color $B_{i}$ by the following periodic pattern: $1, k, 1, k + 1, ..., 1, 2k - 1$. As the period for every
  color different from 1 is $2k$ and the largest used color is $2k - 1$, we conclude that we get a packing coloring of
  $B_{i}$.
 \end{proof}

 \begin{lemma}
  \label{lem-2bands}
  Let $D(1, t)$ be a distance graph, $t \geq 50$, and $B_{i}$, $B_{i + 25}$ its bands. Then it is possible to color
  $B_{i}$ and $B_{i + 25}$ using colors $C = \{1, 18, 19, ..., 35\}$.
 \end{lemma}

 \begin{proof}
  We color the vertices of $B_{i}$ and $B_{i + 25}$ repeating the pattern from the proof of Lemma~\ref{lem-band}. We
  start to color $B_{i}$ at the vertex $i$ and $B_{i + 25}$ at the vertex $i - kt$ for any $k \in \{11, 12,..., 25\}$.
  Lemma~\ref{lem-band} assures that the distance between two vertices colored with color $c$ in a single band is greater
  than~$c$. Let $u \in V(B_{i})$ and $v \in V(B_{i+25})$ be colored by the same color. By the pattern from the proof of
  Lemma~\ref{lem-band} we conclude that the distance between $u$ and $v$ is $\min\{k, 36 - k\} + 25$ which is greater
  than 35. Hence we have a~packing coloring of $B_{i}$ and $B_{i + 25}$.
 \end{proof}

 For a distance graph $D(1, t)$ we use notation $D(1, t) = S_{0}B_{24}S_{25}B_{49}\ldots$ to express that
 we view $D(1, t)$ as a union of strips $S_{0},S_{25},\ldots$ and bands $B_{24},B_{49},\ldots$.

 Now we are ready to prove Theorem~\ref{DistanceGraph}.

 \begin{proof}[Proof of Theorem~\ref{DistanceGraph}]

 \emph{Case 1:} $t$ is odd.\\
  Let $r, s$ be positive integers such that $t = 24s + r$, where $r < 24$ is also odd. Since $t \geq 575$, we get
  $s \geq r$ (for $r = 23$ we have $24s \geq 552$). Thus we have $s$ disjoint strips and $r$ disjoint bands such that
  $D(1, t) = S_{0}B_{24}S_{25}B_{49}... S_{24(r - 1) + r - 1}\zlom B_{24r + r - 1}S_{24r + r}... S_{24(s - 1) + r}$.

  For odd $j = 1, 3, ..., r$, we color the strips $S_{24(j - 1) + j - 1}$ cyclically with the pattern on 24$\times$24
  vertices starting at the vertex $24(j - 1) + j - 1$. For even $j = 2, 4, ..., r - 1$, we color $S_{24(j - 1) + j - 1}$
  cyclically with the pattern on 24$\times$24 vertices starting at the vertex $24(j - 1) + j - 1 - t$. For
  $j = r + 1, r + 2, ..., s$, we color $S_{24(j - 1) + r}$ cyclically with the pattern on 24$\times$24 vertices starting
  at the vertex $24(j - 1) + r - t$. Hence we have a packing coloring of all $s$ disjoint strips of $D(1, t)$ using the
  same principle as in the proof of Lemma~\ref{Strip}.

  For odd $j = 1, 3, ..., r - 2$, we color the bands $B_{24j + j - 1}$ cyclically with the sequence of colors
  $1, 18, 1, 19, ..., 1, 35$ starting at the vertex $24j + j - 1$. For even $j = 2, 4, ..., r - 3$, we color
  $B_{24j + j - 1}$ cyclically with the sequence of co\-lors $1, 18, 1, 19, ..., 1, 35$ starting at the vertex
  $24j + j - 1 - 17t$. We color $B_{24(r - 1) + r - 2}$, $B_{24r + r - 1}$ cyclically with the sequence of colors
  $1, 18, 1, 19, ...\carka 1, 35$ starting at the vertex  $24(r - 1) + r - 2 - 13t$, $24r + r - 1 - 24t$, respectively.
  Hence we have a packing coloring of all $r$ disjoint bands of $D(1, t)$ using the same principle as in the proof of
  Lemma~\ref{lem-2bands}.

  Note that the bands are colored by colors $1, 18, 19, ..., 35$ and the strips are colored by colors $1, 2, ..., 17$
  such that no pair of adjacent vertices is colored with color 1. Then we conclude that we have a packing coloring of
  $D(1, t)$, hence $\chi_{\rho}(D(1, t)) \leq 35$.

  We illustrate this situation on Figure~\ref{fig-oddspiral}. The black vertices are colored by 1 and we color bands
  cyclically
  only with the sequence of colors of length 6 instead of 36 and a strip consists of only 4 bands instead of 24. Note
  that this decomposition is equivalent to our situation.

  \begin{figure}[ht]
   \begin{center}
   \includegraphics[width=12.cm]{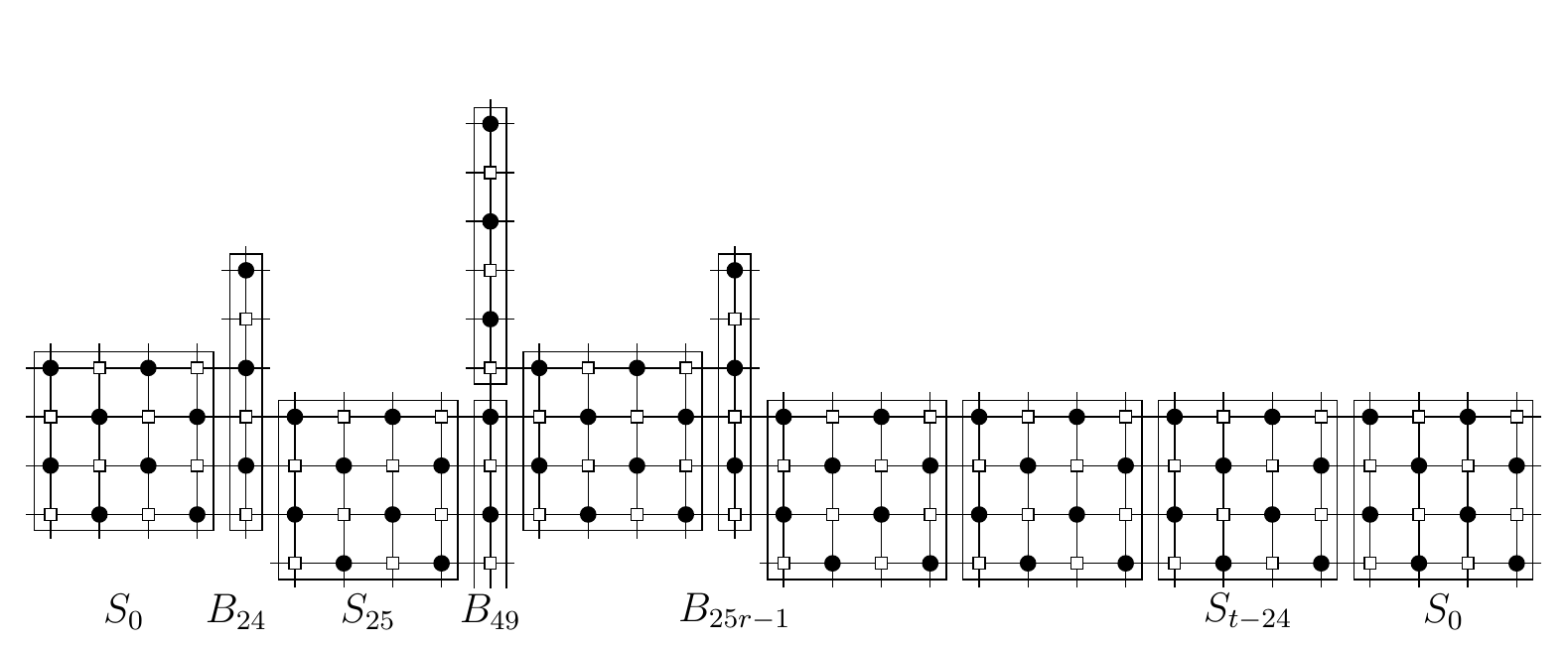}\\
   \end{center}
   \caption{Distance graph $D(1,t)$ for odd $t$.}
   \label{fig-oddspiral}
  \end{figure}

 \medskip

 \emph{Case 2:} $t$ is even.\\
  Let $r, s$ be positive integers such that $t = 24(s + 2) + r$, where $0 < r \leq 24$ is also even. Since $t \geq 648$,
  we get $s \geq r$ (for $r = 24$ we have $24s \geq 576$). Thus we have now $s + 2$ disjoint strips and $r$ disjoint
  bands such that $D(1, t) = S_{0}S_{24}B_{48}S_{49}B_{73}... S_{24(r - 1) + r - 2}B_{24r + r - 2}S_{24r + r - 1}
  S_{24(r + 1) + r - 1}... \zlom S_{24(s + 1) + r - 1}B_{24(s + 2) + r - 1}$.

  For odd $j = 1, 3, ..., r - 1$, we color the strips $S_{0}$, $S_{24j + j - 1}$ cyclically with the pattern on
  24$\times$24 vertices starting at the vertex 0, $24j + j - 1$, respectively. For even $j = 2, ..., r - 2$, we color
  $S_{24j + j - 1}$ cyclically with the pattern on 24$\times$24 vertices starting at the vertex $24j + j - 1 - t$. For
  $j = r, r + 1, ..., s + 2$, we color $S_{24j + r - 1}$ cyclically with the pattern on 24$\times$24 vertices starting at
  the vertex $24j + r - 1 - t$. Hence we have a packing coloring of all $s + 2$ disjoint strips of $D(1, t)$ using the
  same principle as in the proof of Lemma~\ref{Strip}.

  For odd $j = 1, 3, ..., r - 1$, we color the bands $B_{24(j + 1) + j - 1}$ cyclically with the sequence of colors
  $1, 18, 1, 19, ..., 1, 35$ starting at the vertex $24(j + 1) + j - 1$. For even $j = 2, 4, ..., r - 2$, we color
  $B_{24(j + 1) + j - 1}$ cyclically with the sequence of colors $1, 18, 1, 19, ..., 1, 35$ starting at the vertex
  $24(j + 1) + j - 1 - 17t$. We color $B_{24(s + 2) + r - 1}$ with sequence of colors $18, 19, ..., 56$ starting at the
  vertex $24(s + 2) + r - 1$ by Proposition~\ref{Goddard} for $k = 18$. Note the band $B_{24(s + 2) + r - 1}$ is the
  only one with colors greater than 35. We have a packing coloring of all $r$ disjoint bands of $D(1, t)$ by the fact
  that the distance between an arbitrary vertex of $B_{24(s + 2) + r - 1}$ and a vertex of any other band is at least 49
  and using the same principle as in the proof of Lemma~\ref{lem-2bands}.

  Note that the bands are colored by colors $1, 18, 19, ..., 56$ and the strips are colored by colors $1, 2, ..., 17$
  such that no pair of adjacent vertices is colored with color 1. Then we conclude that we have a packing coloring of
  $D(1, t)$, hence $\chi_{\rho}(D(1, t)) \leq 56$.

  We illustrate this situation on Figure~\ref{fig-evenspiral}. Note that this decomposition is equivalent to our
  situation as in Case 1.

  \begin{figure}[ht]
   \begin{center}
   \includegraphics[width=12.cm]{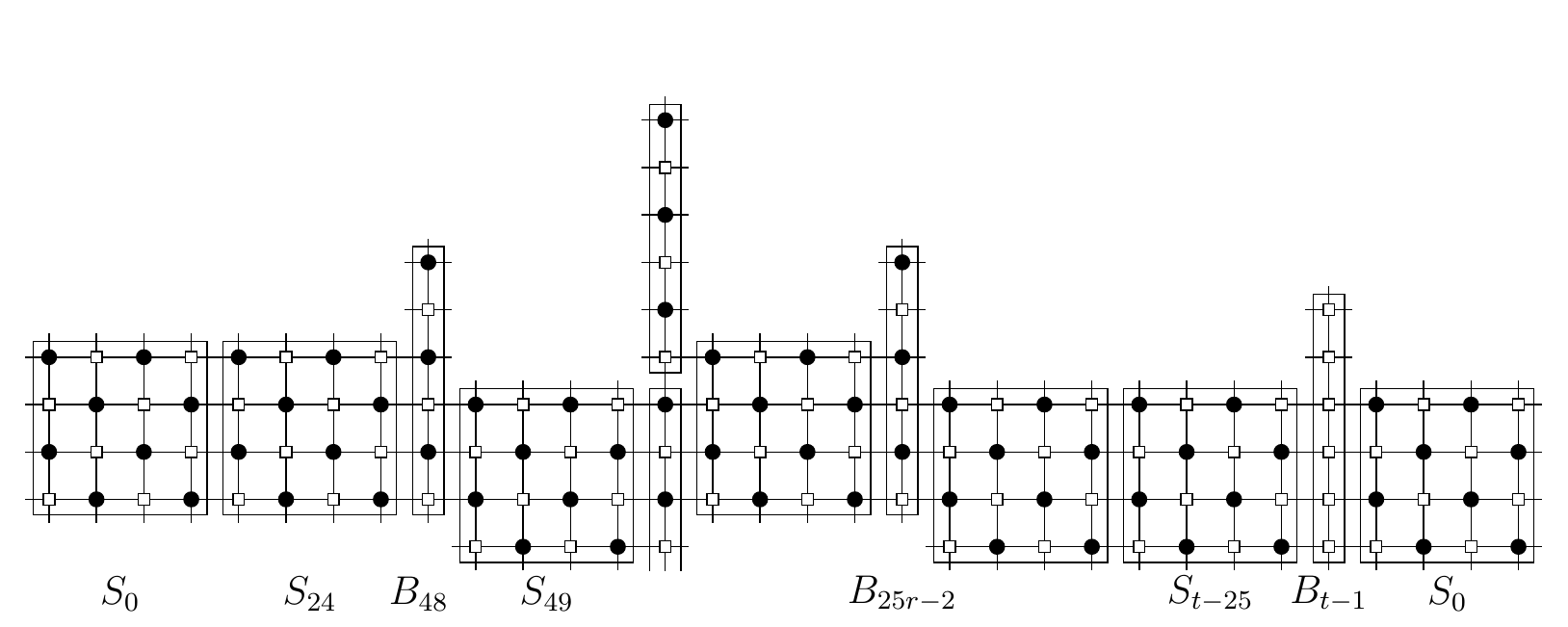}\\
   \end{center}
   \caption{Distance graph $D(1,t)$ for even $t$.}
   \label{fig-evenspiral}
  \end{figure}
\end{proof}

 \begin{table}[ht]
  \begin{center}
  \begin{tabular}{|c|c|c|c|c|}
   \hline
   $r$ & $t \geq$ &   & $r$ & $t \geq$ \\
   \hline
    1 &  25 &   &  2 &  98 \\
    3 &  75 &   &  4 & 148 \\
    5 & 125 &   &  6 & 198 \\
    7 & 175 &   &  8 & 248 \\
    9 & 225 &   & 10 & 298 \\
   11 & 275 &   & 12 & 348 \\
   13 & 325 &   & 14 & 398 \\
   15 & 375 &   & 16 & 448 \\
   17 & 425 &   & 18 & 498 \\
   19 & 475 &   & 20 & 548 \\
   21 & 525 &   & 22 & 598 \\
   23 & 575 &   & 24 & 648 \\
   \hline
 \end{tabular}
 \end{center}
   \caption{Table for $t$ depending on $r$.}
   \label{Ruzne r}
  \end{table}

 Note that in some cases we can decrease $t$ for which Theorem~\ref{DistanceGraph} is true. It depends on $r$ from
 the proof of Theorem~\ref{DistanceGraph}. We have $t \geq 24r + r$ for odd $t$ and $t \geq 24r + r + 48$ for even $t$
 (see Table~\ref{Ruzne r}).

 \section{Lower bound from square lattice}

 In this section we give a proof of the lower bound for $\chi_{\rho}(D(1, t))$.

 \begin{proof}[Proof of Corollary~\ref{Dolni mez}.]
 By the proof of Theorem \ref{CtvercovaMrizka}, a finite square lattice $15 \times 9$ cannot be colored using $11$
 colors. Clearly $D(1,t)$ contains a finite square grid as a subgraph and $t \geq 9$ assures existence of the square
 lattice $15 \times 9$ in $D(1,t)$. Therefore,  $\chi_{\rho}(D(1, t)) \geq 12$ for every $t \geq 9$.
 \end{proof}

 \section{Conclusion}
  We have shown that the packing chromatic number of an infinite distance graph $D(1, t)$ is at least 12 for $t \geq 9$
  and at most 35 for odd $t$ greater or equal than 575 or at most 56 for even $t$ greater or equal than 648. Moreover, we
  have found some smaller values of $t$ for which Theorem \ref{DistanceGraph} holds. The next research in this area can
  be focused on finding better bounds for $D(1, t)$. In particular, obtaining a~lower bound for $D(1,t)$ which would
  exceed the upper bound for the square lattice would be an interesting result.

 \section{Acknowledgment}
  We would like to express our thanks to Ji\v{r}\'{\i} Fiala for comments and fruitful discussion.

  The access to the METACentrum computing facilities, provided under the programme "Projects of Large Infrastructure for
  Research, Development and Innovations" LM2010005 funded by the Ministry of Education, Youth and Sports of the Czech
  Republic, is highly appreciated.

\end{document}